\documentclass{jams-l}
\usepackage{booktabs}

\newtheorem{theorem}{Theorem}[section]
\newtheorem{lemma}[theorem]{Lemma}

\theoremstyle{definition}
\newtheorem{definition}[theorem]{Definition}

\theoremstyle{remark}

\numberwithin{equation}{section}

\begin{document}

\title{The Linear Correlation of $P$ and $NP$}


\author{Bojin Zheng}
\address{South-Central Minzu University, Wuhan, Hubei,430074, China}
\curraddr{182 Minzu Road, Hongshan District,Wuhan, Hubei,430074, China}
\email{zhengbj@mail.scuec.edu.cn}

\author{Weiwu Wang}
\address{China Development Bank, Beijing 100080, China}
\curraddr{NO.16 Taipingqiao Street, Xicheng District, Beijing, 100080, China}
\email{wangweiwu@cdb.cn}

\subjclass[2010]{Primary }

\date{2023-2-22}

\dedicatory{}

\begin{abstract}
$P \overset{\text{?}}{=} NP$ or $P\ vs\ NP$ is the core problem in computational complexity theory. In this paper, we proposed a definition of linear correlation of derived matrix and system, and discussed the linear correlation of $P$ and $NP$. We draw a conclusion that $P$ is linearly dependent and  there exists $NP$ which is is linearly independent and take a 3SAT instance which belongs to $NP$ as the example , that is, $P \neq NP$.  
\end{abstract}

\maketitle


\section{Introduction the $P$ and $NP$}

Now the $P\ vs\ NP$ problem or the $P \overset{\text{? }}{=} NP$ is the most important problem in computer science. It was proposed by Stephen Cook in 1971 when he published his paper \textit{The complexity of theorem-proving procedures} \cite{Cook71thecomplexity}. In 2000, the question was selected as one of the seven Millennium Prize Questions. For decades, a large number of researchers have developed various technologies to try to solve this problem\cite{luo2015en}, but have not yet solved it completely.

Colloquially, $P$ refers to the set of problems solvable by a Deterministic Turing machine (DTM or TM) in polynomial time, and $NP$ refers to the set of problems verifiable by a Deterministic Turing machine in polynomial time. $P\ vs\ NP$ problem or $P \overset{\text{? }}{=} NP$ problem means: whether any problem that can be efficiently verified  can be efficiently solved or not. Cook proposed $NP-$completeness \cite{Cook71thecomplexity} and described the most difficult $NP-$complete problems in $NP$. $NP-$complete problems can be reduced to $3SAT$ problem, that is, if one $NP-$complete problem can be solved efficiently, then all $NP-$complete problems can be solved efficiently. At present, thousands of  $NP-$complete problems have been found in practice, covering dozens of disciplines including biology, chemistry, economy etc. To solve $P \overset{\text{? }}{=} NP$ problem has not only important practical significance, but also important scientific and philosophical significance.

\section{Basic Strategy to Solve $P$ vs $NP$}

The definition of $NP$ bases upon the Non-Deterministic Turing machine (NDTM), which can be regarded as a foreseen, polynomial-time algorithm. Because $P \subseteq NP$, $P \neq NP$ actually means $P \subset NP$, or say, there is at least one element in $NP$ which does not belong to $P$, i.e., which can not be resolved by a general, unforeseen, polynomial-time algorithm.

Some researchers insist on that any statement on proving $P \neq NP$ should prove that at least an $NP$ can not be resolved by any unforeseen, polynomial-time algorithm, or say, a general, unforeseen, polynomial-time algorithm in advance. Actually, this belief is not only dubious, but also unnecessary.

When we talk about ``a general, unforeseen, polynomial-time algorithm'', we denote it ``the myth algorithm'' here, we have actually assumed its existence. At the least, we have endowed it one name or phrase, that is, it exists in the meaning of grammar. Of course, it exists formally. From the viewpoint of formal logic, we can not prove its non-existence. However, our task has been aimed to prove the non-existence of myth algorithm. There comes a logic contradiction.

Fundamentally,  we can not prove the non-existence of the myth algorithm before proving $P \neq NP$. The argument here is not our invention. Historically, Aristotle discussed the non-existence of vacuum with the similar argument. Galois faced the same dilemma and found a way to prove the non-existence of a general algebraic methods for quintic equations. According to the case of Galois, the myth algorithm is unnecessary to discuss. By the way, the existence of the phrase ``the myth algorithm'' has assumed that $P = NP$, and can not be used to prove $P \neq NP$.

A common method to prove the non-existence is based on the feature, or property of existence. For example, if we have $\forall x \in X, f(x) =1$ and $ \exists y , f(y) = 0$, obviously,  we can have $y \not \in X$. That is, if we have $\forall p \in P, LinearDep(p) = true$, and $ \exists np \in NP, LinearDep(np) = false$, we will have $np \not \in P$, i.e., $P \neq NP$.  

In fact, the predict $LinearDep(x)$ means the linear dependence of $x$. That is, if we can prove that for any instance of $P$, it is linearly dependent,  and for an instance of $NP$, it is linearly independent, then we have proved $P \neq NP$. Noticing that $3SAT \in NP$, our proof will reduce to two propositions:
\begin{enumerate}
	\item  $\forall p \in P, LinearDep(p) = true$
	\item $\exists 3sat, LinearDep(3sat) = false$ 
\end{enumerate}

Of course, such two propositions are also feasible:
\begin{enumerate}
	\item  $\forall p \in P, LinearDep(p) = true$
    \item $\exists np \in NP, LinearDep(np) = false$
\end{enumerate}

\section{Definitions and Basic Concepts}

\subsection{The Discrete representation of computable problems}

\begin{definition}{\bf The Discrete representation of computable problems. }
	For any given computable problem, it can be represented by a set of the pair of input and output, i.e. $<X,Y>$.	
\end{definition}

\begin{theorem}\label{eqiv.pa}
	There is a one-to-one mapping relationship between any computable problem and its time-optimal  algorithm.
\end{theorem}

\begin{proof}
	First to prove that any computable problem must correspond to an time-optimal algorithm. For any computable problem, there are a series of algorithms to solve the problem. For any algorithm, there is a measure of solving time. Since the measure value is full-order relationship, there must be an time-optimal algorithm.
	Then to prove that any time-optimal  algorithm must correspond to a computable problem. For the time-optimal  algorithm, the input and output of the algorithm are written as a tuple $<X,Y>$, which must correspond to the original problem.
	Therefore, computable problem and time-optimal algorithm are one-to-one mapping relationship, that is, equivalence relationship.
\end{proof}

As to the theorem\ref{eqiv.pa}, Oded Goldreich talked about the same content\cite{Gold2015}: ``... in the theory of computation, where the representation of objects plays a central role. In a sense, a computation merely transforms one representation of an object to another representation of the same object. ... Indeed, the answer to any fully specified question is implicit in the question itself, and computation is employed to make this answer explicit.''

According to theorem \ref{eqiv.pa}, in this paper, the concept of the $P$ problems is used the same as the polynomial deterministic algorithm, that is, the $P$ algorithm (denoted as $P_f$), and the concept of the $NP$ problems is used the same as the deterministic algorithm for the $NP$ problems.

\subsection{The Philosophy of Computation and Representation}

According to theorem\ref{eqiv.pa} and opinion of Oded, we can treat the problem and the algorithm as two sides of one coin. As to the coin, we name it ``system''. That is, the problem is a representation of the system, and the algorithm is another. We denote the system as a tuple $<X, f(x), Y>$. Obviously, this tuple can be represented as a matrix, whose rows indicate the records or instances of the tuple, we call this matrix as derived matrix. 

The (time-optimal) algorithm, the problem and the system can be regarded as the same thing, and they share the same derived matrix. Here, we call the derived matrix the discrete representation of the system, which can be rewritten as a tuple $<X, f(x), Y>$, i.e., the systemic representation.

As to the algorithm, it can be regarded as the algorithmic representation because it is equivalent to the system and the problem.

Considering to any algorithm can be understood as the Turing machine, the corresponding Turing machine will be called as Turing machine representation. If there exists equivalent abstract machine, the corresponding machine can be called as machine representation. For example, the Random Access Machine.  

Because the universal Turing machine exists, a Turing machine can be coded as an integer (program code), i.e., number, we can use an integer to represent the problem. That is, the problem has a number field representation.

Because the process of execution of a Turing machine can be modeled as a graph, whose nodes indicates the states, all the executions can be organized a tree, the root is the HALT status and the leaves are the initial statuses. This kind of method call as topological representation. 

Of course, the SAT problem can be regarded as a kind of representation based on symbolic logic, we can call it the SAT representation or the symbolic logic representation. 

When we think about the nature of $P$ and $NP$ from the view of representation, some interesting questions will arise. For examples, what will happen when the number field representation is extended to complex style? How to extend the Random Access Machine to deal with $NP$ in polynomial time?  What is the difference between the topological graph of $P$ and of $NP$? How to extended the systemic representation to represent $NP$? and so on. This paper will not answer all these questions, leaving some of them to the future work.

\subsection{The Linear Correlation of Tuple}

Assuming that the natural number $S$ can be expressed as a string of symbols, $S=S_n\cdots S_2S_1$, here $S_i \in \{0,1\}$, then $S$ can be regarded as a bit vector with dimension $n$. In this paper, vertical vectors are used, since that vertical vectors are easy to correspond to representations of natural numbers.

Considering the tuple $\Phi = < X, Y > $, let $X = S = S_n \cdots S_2S_1 $, $Y \in \{0, 1 \} $, $< X, Y > $ will form a new $(n + 1)$-dimension vector $\Phi_A = (S_n, \cdots, S_2, S_1, Y) $.

$\Phi \subseteq X \times Y$. Here, $\times$ means Descartes product.

When $\Phi$ is linearly dependent, and $X$ is linearly independent, $\Phi$'s dimension number will be equivalent to $n$.

Because  we shall discuss the $P$ and $NP$, whose input are natural number, the dimension number of $X$, i.e., the binary representation of natural number, is variant relative to $n$, we just treat $X$ a whole, and let its dimension number to be 1. The same to $Y$, its dimension number is 1. We only focus on the linear correlation between $X$ and $Y$.

\begin{definition}\label{linearRela}{\bf The linearly dependent of a tuple.}
	If a vector set $\Phi_A$ is linearly dependent, then the tuple $\Phi$ is linearly dependent.
	If a vector set $\Phi_A$ is linearly independent, then the tuple $\Phi$ is linearly independent.
\end{definition}	

According to the definition \ref {linearRela}, it is clear that, if the vector set of $< X, Y > $ is linearly dependent, then the dimension number of $< X, Y > $ is 1, noted as $|| < X, Y > || = 1 $. If  linearly independent, the dimension number of the tuple $< X, Y > $  is 2, noted as $|| < X, Y > || = 2 $.

\begin{theorem}\label{theorem.cup.linerRela}
	Assume that there are two sets $A$ and $B$, $X$ and $Y$ are linearly independent separately, $A, B \subseteq  X \times Y $,and $A \cap B = \emptyset$. Let $||A||=1$,$||B||=2$, then $||A \cup B||=2$. Let $||A||=1$, $||B||=1$, then $||A \cup B||$ is uncertain, maybe 1, maybe 2. Let $||A||=2$, $||B||=2$, then $||A \cup B||=2$.
\end{theorem}

Theorem \ref{theorem.cup.linerRela} shows that the union of two linearly dependent tuple may be linearly dependent or independent; The union of two linearly independent tuple must be linearly independent; If one tuple is linearly dependent and the other tuple is linearly independent, then the union of these two tuples must be linearly independent.

Conversely, a linearly dependent proper subset may be found in a linearly independent tuple; A linearly independent proper subset cannot be found in a linearly dependent tuple, but only a linearly dependent proper subset can be found.

\subsection{Examples of Linear Correlation of Tuple}

Similar to the terms in the computer database, we call the instances of tuple as records. 

The record set $\{<x_0,y_0>\}$ is always linearly dependent when there is only one element. Here, $x_0 \in X$ and $y_0 \in Y$.
 
When $< x_0 ,y_0 > = < 1, 1 > $, $k_1 = 1 $, $k_2 = -1 $, have $k_1 * 1 + k_2 *(-1) = 0 $, therefore, it is linearly dependent.

The record set $\{<0,0>,<1,1> \}$ is linearly dependent.

For the record set $\{<1,0>,<0,1> \}$, it can be calculated that the set is linearly independent. Therefore, $\{< 1, 0 >, < 0, 1 >, < 0, 0 > \} $ and $\{< 1, 0 >, < 0, 1 >, < 0, 0 >, < 1, 1 > \} $ is linearly independent.

For record set $\{< 1, 0 >, < 0, 1 > \} $, if  $y '= \bar {y} $, then we get $\{< 1, 1 >, < 0, 0 > \} $, at this time it is linearly dependent. This fact shows that in some cases, the inverse of a vector element may lead to the change of its linear correlation.

But $\{< 1, 0 >, < 0, 1 >, < 1, 1 > \} $ is linearly independent, the inversely transformed result $\{< 1, 1 >, < 0, 0 >, < 1, 0 > \} $ is also linearly independent.

\subsection{The Linear Correlation of System}

The system can be written as $<X, f, Y> $. Here, $f:X \rightarrow Y, X= \mathbb{N}, Y=\{0,1\}$; because $f$ can be represented as a Turing Machine, its code is a binary constant.  Because $f$ has a result, $f$ often can be rewritten as $<f_{code}, f(x)>$. That is, $<X, f, Y> $ can also be rewritten as $<X, f_{code}, f(x), Y>$. $f_{code}$ can be transferred to a program under a certain universal Turing Machine.

\subsection{The Definition of Derived Matrix of System}

Generally, $P_f$ can be expressed as 
$P_f:X \rightarrow Y, X=x_1 x_2\cdots x_n  \subseteq \mathbb{N}, x_i \in \{0,1\},Y=\{0,1\}, i \in \{1,n\}, Time(P_f) \in Poly(n)$.

The derived matrix of $P_f$, or say, of the system $<X, P_f, P_f(x),Y> $ can be defined as 
 \begin{equation}\label{phi}
 	\Phi(P_f) = (x_1,x_2,\cdots,x_n,P_f, P_f(x), Y)
 \end{equation}

Here, $P_f$ is a program code, i.e., a constant. 

\subsection{The Impact of the Constant in the Input}

In this paper, the linear correlation is the focus. As to the linear correlation of the derived matrix, a trivial problem is: how about when the input has a constant?  

Actually, when we defined the linear correlation of tuple, we focus on the relationship between $Y$ and $P_f(x)$. When $Y= P_f(x)$, then linearly dependent; When $Y \not = P_f(x)$, we should check the linear independence according to the definition of linear space. But when $Y \not = P_f(x)$ and $Y= P_f(x)$  for some $x$s, the system will be linearly independent. 

As to the constant $P_f $ in the derived matrix, we should ignore it and do no calculate the linear correlation with it.    

\subsection{The Theorem of Linear Dependence of System}
Because $f(x)=Y$, $<X, f_{code}, f(x), Y>$ will be linearly dependent. In fact, $Y$ always equals to $f(x)$, the system will always be linearly dependent. 

Obviously, $P_f$, the function of $\forall p \in P$, will have a corresponding system which is linearly dependent, because $P_f$ is a $f$. Therefore, we have $LinearDep(p) =true$;
\begin{theorem}\label{p.lindep}
	$\forall P_f, LinearDep(P_f) =true$
\end{theorem}
\begin{proof}
	Omitted.
\end{proof}

\section{The Details of Proving $P \neq NP$}

According to theorem \ref{p.lindep}, arbitrary $P$ is linearly dependent. If we can take a $NP$ example which is linearly independent, then we can prove $P \neq NP$. We will illustrate the nature of $NP$ and prove that some $NP$s are linearly independent.

\subsection{Extend the System to the Complex Boolean System}

\begin{definition}{\bf The complex boolean number.}
	We call the number with the expression  $z = a + bi, a \in \{0,1\},b \in \{0,1\}$ the complex boolean variable , the real part $\theta(z)=a$.  $0+0*i=0$.	
\end{definition}

\begin{definition}{\bf The value set of Complex Boolean variable.}
	The set of values of complex Boolean numbers is called the complex Boolean set. Noted as $CB=\{0, i, 1, 1+i\}$.	
\end{definition}

\begin{definition}{\bf The modulus of complex Boolean numbers.}
	A Boolean number $z=a + bi$, its modulus $\psi (z)$ of $z$ has,
	\begin{equation}\label{bool.modulus}
		\psi (z) = a \vee b
	\end{equation}
\end{definition}

According to the definition of the system, obviously, the input is based upon the real integer which can be treated as boolean variables. As has been known, every number can be represented with 4 as the carry base. That is, we can combine every 2 bits as a whole, such 2 bits can be treated as complex boolean variables, one bit is real part, and the other bit is imaginary part.
 
As to the value of $f(x)$, we only focus on the real part, therefore, we can use $\theta(f(x)) $ to replace previous $f(x)$. As to $Y$, because $Y \ in CB$, i.e., is a complex boolean variable, we can use $\psi(Y)$ to replace previous $Y$.
Thus, the complex boolean system can be rewritten as $<X, f, \theta(f(x)), \psi(Y)>$.

\subsection{$NP$ is the Complex Boolean System}
Cook proposed $NDTM$ to describe $NP$. Here, we prove that any $NDTM$ recognizable language $L$ can be converted to the $CNF$ formula on complex Boolean variables. For the sake of description, we give the definition of Turing machine and non-deterministic Turing machine.

\begin{definition}[Turing Machine]
	A Turing Machine $M$ is a tuple $\left\langle \Sigma, \Gamma, Q, \delta \right\rangle $. $ \Sigma $ is a finite nonempty set with input symbols; $ \Gamma $ is a finite nonempty set, including a blank symbol $ b $ and $ \Sigma $; $ Q $ is a set of possible states, $ q_{0} $ is the initial state, $ q_{accept} $ is an accept state, $ q_{reject} $ is a reject state; $ \delta $ is a transition function, satisfying
	\[ \delta:(Q - \left\lbrace q_{accept}, q_{reject}\right\rbrace ) \times \Gamma \to Q \times \Gamma \times \left\lbrace -1, 1 \right\rbrace  \]
	if $ \delta(q,s)=(q',s',h) $, the interpretation is that, if $M$ is in state $q$ scanning the symbol $s$, then $q'$ is the new state, $s'$ is the symbol printed, and the tape head moves left or right one square depending on whether $h$ is $-1$ or $1$. $ C $ is a configuration of $ M $, $ C=xqy $, $ x,y \in \Gamma^{*} $,$ y $ is not empty, $ q \in Q $. The computation of $ M $ is the unique sequence $ C_1, C_2, \cdots $. If the computation is finite, then the number of steps is the number of configurations.		
\end{definition}

\begin{definition}[Nondeterministic Turing Machine]
	A Nondeterministic Turing Machine $N$ is a tuple $ \left\langle \Sigma, \Gamma, Q,\delta_1, \delta_2 \right\rangle $. The difference between $ N $ and $ M $ lies in the transition function. There are two functions $ \delta_1 $ and $ \delta_2 $ in $ N $. In any configuration $ C_i $, you can choose between the different functions. The choice is denoted as $ \Delta_{i,j} $, $ j \in \{1,2\} $. The sequence of choices is a   path of computation. If there exits a path $ \Delta_{1,j_1}, \Delta_{2,j_2}, \cdots  $ leading to the accepting state, $ N $ halts  and accepts. The number of computation is the number of steps in the shortest path.
\end{definition}

There is a lemma to compute the size of $ CNF $ \cite{luo2015en}.


\begin{lemma}
	For a given boolean function $ f:\{ 0,1\}^l \to \{0,1\}$, there exists a $ CNF $ formula $ \varphi $, which size is $ l2^l $, satisfying
	$ \varphi(u)=f(u) $ , $ \forall u \in\{0,1\}^l $. The size of a $ CNF $ formula is the number of $ \wedge $ and $ \vee $.
\end{lemma}

\begin{theorem}
	If a language $ L $ is accepted by a Nondeterministic Turing Machine $ N $, then  $ L $ is reducible to a $ CNF $ formula $ \varphi_z $ of Complex Boolean $ z $.
\end{theorem}
\begin{proof}
	Suppose a Nondeterministic Turing Machine $ N $ accepts a language $ L $ within time $ T(n) $. There xists a path $ \Delta_{1,j_1}, \Delta_{2,j_2}, \cdots, \Delta_{T(n),j} $ leading to the accepting state.
	
	Suppose $\delta_0'= \delta_1 \cap \delta_2 $, $ \delta_1' = \delta_1 - \delta_0' $, $  \delta_2' = \delta_2 - \delta_0' $,the path $ \Delta_{1,j_1}, \Delta_{2,j_2}, \cdots $ changes to the path $ \Delta_{1,j_1'}', \Delta_{2,j_2'}', \cdots $, where $ \Delta_{i,0}' $ is deterministic, $ \Delta_{i,1}' $ and $ \Delta_{i,2}' $ are nondeterministic. 
	
	For example, a path $ \Delta_{1,0}', \Delta_{2,1}', \Delta_{3,2}', \Delta_{4,0}', \cdots $, $ \Delta_{1,0}' $ and $ \Delta_{4,0}' $ are deterministic, which are both from $\delta_0' $; $ \Delta_{2,1}' $ and $ \Delta_{3,2}' $ are nondeterministic, which are  from $\delta_1' $ and $\delta_2' $.
	
	Given an input $ x \in \left\lbrace 0,1\right\rbrace^{n}  $, we will construct a proposition formula $ \varphi_z $ such that $\varphi_z $ is satisfiable iff $ N $ accepts $ x $.
	
	Step 1 construct formula $ F_1 $ to assert the input $ x $. Suppose $ y_1, y_2, \cdots, y_i, $ $\cdots, y_n $, $ y_i = x_i $ . The size of $ F_1 $ is $ 4n $. The number of variables is $ n $. $F_1$ exits in $a$ of $z$.
	
	Step 2 construct formula $ F_2 $ to assert the initial configuration $ C_0 $. A configuration can be coded in a length of $ c $ bits string $w_0$. $ c $ is decided by $ \left| Q\right|  $ and $ \left| \Gamma \right| $. The size of $ F_2 $ is $ c2^c $. The number of variables is $ c $. $F_2$ also exits in $a$ of $z$. 
	
	Step 3 construct formula $ F_3 $ to assert $ 1\leq i \leq T(n)$ configurations. $ C_i $ (coded by $w_i$) is decided by the state of $ C_{i-1} $, the input symbol $ x_{inputpos(i)} $, the tape symbol (which is decided by the last operation of the head $ C_{prev(i)} $). There exists a function $ F $ (which is decided by $ \delta_0' $, $ \delta_1' $ and $ \delta_2' $), satisfying 
	\[ C_i = F(C_{i-1}, x_{inputpos(i)}, C_{prev(i)}) \]
	For each configuration $ C_i $, the number of variables is at most
	$ 3c+1 $. The size of formula is $ (3c+1)2^{(3c+1)}$. For all of configurations, the size of formula is $ T(n)(3c+1)2^{(3c+1)} $ at most.
	
	If $ C_i $ is in the path of $ \Delta_{i,j}'=\Delta_{i,0}' $,then $ F $ is decided by $ \delta_0'$. $ F_3 $ exists in $ a $ of $ z $.
	
	If $ C_i $ is in the path of $ \Delta_{i,j}'=\Delta_{i,1}' $ or $ \Delta_{i,2}' $, then $ F $ is decided by $ \delta_1' $ or $ \delta_2' $. $ F_3 $ exists in $ b $ of $ z $.
	
	Step 4 construct formula $ F_4 $ to assert the final state and output. It is similar to Step 2. 
	
	At last, we AND all part of formulas. It becomes $ z = \varphi_z $. 
	
	Finally, the formula asserts that $ N $ reaches an accepting state iff $a=1$ of $ \varphi_z $ and $b=1$ of $\varphi_z$.
	
\end{proof}	


The construction of $ \varphi_z $ is a computation path of $ N $. So, $ \varphi_z $ can be computed by $ N $.

\subsection{The Derived Matrix in the Complex Style}

The derived matrix of $NP_f$, or say, of the system $<X, NP_f, \theta(
NP_f(x)), \psi(Y)> $ can be defined as 
\begin{equation}\label{phi.np}
	\Phi(NP_f) = (a_1,b_1, a_2, b_2,\cdots,a_n,b_n, NP_f, \theta(NP_f(x)), \psi(Y))
\end{equation}

Here, $NP_f$ is a program code, i.e., a constant. 

Obviously, when the $X$ is treated as real integer, $\Phi(NP_f) = \Phi(P_f)$. That is, the derived matrix of $P$ is a special case of $NP$.

Because we do not care about the $NP_f$ and $P_f$ as constants in $\Phi()$, we can rewrite the derived matrix as 
\begin{equation}\label{phi.np.2}
	\Phi'(NP_f) = (a_1,b_1, a_2, b_2,\cdots,a_n,b_n, \theta(NP_f(x)), \psi(Y))
\end{equation}

\subsection{The Linear Independence of the Derived Matrix}

\begin{theorem}\label{lin.indep}
	\item $\exists np \in NP, LinearDep(np) = false$	 
\end{theorem}

\begin{proof}
	Assume for $x_0, np_f(x_0) = a_0 + i*b_0$. When $b_0 =0$, $\theta(np_f(x))=\psi(a_0+ i* b_0)$.When $b_0 =1, a_0=0$, $\theta(np_f(x))=0 \not =\psi(a_0+ i* b_0)=1$. Therefore, $LinearDep(np) = false$.	
\end{proof}

\begin{theorem}
	$P \not = NP$.
\end{theorem}

\begin{proof}
	According to Theorem \ref{p.lindep} and Theorem \ref{lin.indep}, $P \not = NP$.
\end{proof}

\section{An Example and its Explanation}

We will take a 3SAT example to illustrate the validation of $P \not = NP$.

\subsection{The definition of SAT problems}

The SAT problem is also called the Boolean satisfiability problem. This question can be expressed as, given a conjunctive normal form (CNF), that is, a series of disjunctive forms of the clause of the conjunctive form, ask whether there is a set of assignments that make the whole expression true.

If there is an assignment that satisfies the condition, the expression is satisfied; If no such assignment exists, it is not satisfied.

Obviously, if we want the whole CNF to be true, we need every clause to be true. And inside each clause it is a disjunctive form, so at least one of the literals needs to be true for this clause to be true.

Formally, the problem can be stated as:

Assume that there is a Boolean vector $x=(x_1,x_2, \cdots, x_n) $, and a conjunctive normal form $f $, $f(x)=C_1 \wedge C_2 \cdots  \wedge C_m$, among them, the $C_{\textit l }=\vee_1^k(x_j), 1 \le {\textit l } \le m, 1 \le j \le n$, if there exists $x$  such that $f(x)=1$, output $x$, otherwise, output 0. That is, the result set is $y=(y_1,y_2)$, where $y_1=x,y_2=f(x)$.

When $k=2$, it is called the $2SAT$ problem, and when $k=3$, it is called the $3SAT$ problem.

\subsection{Tarjan's Algorithm Solving $2SAT$ Problems}\label{tarjan}

Tarjan's algorithm is a very efficient algorithm for the $2SAT$ problem.

The first step in Tarjan's algorithm is to construct the graph. The method is: set up $2n$ node, respectively $x_1, x_2, \cdots, x_n $ and $\bar{x}_1, \bar{x}_2, \cdots, \bar{x}_n $.
Then we convert each clause into the form of $u \rightarrow v$, make a $u \rightarrow v$ line in the graph, and connect $\bar{v} \rightarrow \bar{u}$ according to the inverse theorem.

The second step of Tarjan's algorithm is to determine the strongly connected component.

In the directional graph $G$, two vertices are said to be strongly connected if there is a directed path from $u$ to $v$, and also a directed path from $v$ to $u$.
$G$ is said to be a strongly connected graph if every two vertices of the directed graph $G$ are strongly connected.

A formula is a contradiction when any node corresponding to the positive and negative literals of $x_i$ appears on a strongly connected path. Because on a strongly connected path, $x_i \rightarrow \bar{x}_i$, and have,$\bar{x}_i \rightarrow x_i$, form a cycle. Therefore, $\bar{x}_i$ and $x_i$ appear in the cycle at the same time, which constitutes a contradiction.

A strongly connected graph is a linearly dependent vector. Because any two nodes in the strongly connected graph are strongly connected, according to the definition of strong connection, there is $u \rightarrow v$ and $v \rightarrow u$, that is, $u \equiv v$ and $v$ are sufficient and necessary conditions for each other. Therefore, $u \equiv v$ , that is, they are linearly dependent.

For a strongly connected graph, you can make a contraction, that is, replace all variables in a strongly connected graph with one variable.

The maximal strongly connected sub-graph of  directed non-strongly connected graph is called strongly connected component. That is, to find a maximal sub-graph that is strongly connected in a graph that is not a strongly connected graph, or say, linearly independent vector when there is no contradiction. This linearly independent vector is called a strongly connected component.

The third step of Tarjan's algorithm is the assignment.

The assignment is actually to solve a linear Boolean equation. Boolean equations are easy to solve.

After excluding the linearly independent  expression $f$ with Tarjan's algorithm, there must be,

\begin{equation}\label{key111}
	k_1*x_1 + k_2 *x_2 + \cdots + k_n * x_n - k_{n+1} *f =1
\end{equation}

Let $q(x)= k_1*x_1 + k_2 *x2 + \cdots + k_n * x_n$,
then we have,
$q \in [0,2]$.

If $q=0$, then $x_1=x_2= \cdots = x_n=0$.

If $q=1$,then select a $i \in [1,n]$,  let $x_i=1$,the other bits are set to 0, and form $x$, hence valid this $x$, so only $n$ validation is needed.

If $q=2$,then select  $i \in [1,n]$ and $j \in [1,n]$, $i \neq j$,  let $x_i=1$ and $x_j=1$, the other bits are set to 0,and form $x$, hence valid this $x$, so only $n*(n-1)/2$ validation are needed.

\subsection{the explanation on SAT problem and linear dependence }\label{axioms}
In this section, we will explain the essential difference in linear correlation between $2SAT$ problem and $3SAT$ problem. The essential difference comes from the conflict of three axioms or the abandon of the law of exclusion in binary logic. If $P=NP$, that is, the algorithm for $2SAT$ can be used to solve the $3SAT$ problem, it must lead to the conflict of three important axioms in algebra and computation, or the abandon of the law of exclusion in binary logic.

\subsection{Three axioms in computation and algebra}
In computer science, it is a basic axiom that computational steps are performed step by step in sequence. In algebra, there are two axioms, namely the commutative law of addition and the associative law of addition. Here it is stated as follows.

\begin{definition}{\bf the sequential execution axiom.} In any trusted computing model, the algorithm's execution steps satisfy the full-order relation.
\end{definition}

\begin{definition}{\bf commutative law of addition.} For addition, always has $a + b = b + a$.
\end{definition}

\begin {definition} {\bf associative law of addition.} For addition, always has $(a + b) + c = a + (b + c) $.
\end{definition}

Obviously, the logic OR operation $\vee$ is an addition operation.

For the sequential execution axiom, in the parallel computing scenario, there is an equivalent partial order relation of the instruction execution steps for some algorithms, but in essence, is still full order relation. Deterministic Turing machines satisfy this axiom.

For the sake of discussion, only the commutative and associative laws of addition are mentioned here, and the commutative and associative laws of multiplication are omitted.
In fact, since the axiom of sequential execution of computation is naturally proposed in essence, addition and multiplication are not equal under the commutative and associative laws.

\subsection{ Linear Dependence of $2SAT$ problems}

The $CNF$ clause in the $2SAT$ problem allows at most two variables and, if there is a 2-ary operator, at most one 2-ary operator. That is, in the $2SAT$ problem, the associative law of logic OR operating $\vee$ does not exist. In the case of conjunctive normal form, every clause is required to be true. Therefore, the commutative law and associative law of each clause under logic AND operation do not conflict with the sequential instruction execution axiom of computation.

Depending on the discrete representation, $2SAT$ can be written as $<<X,f>,Y>$, where $<X,f>$ is the input to $2SAT$. That is,

\begin{equation}\label{2SAT.1}
2SAT=<<X,f>,Y>
\end{equation}

For any computable $f$, $f$ can be written as $<X,Y>$. Therefore, there is,

\begin{equation}\label{2sat.2}
2SAT=<<X,<X,Y>>,Y>=<X,Y>
\end{equation}

Since $Y$ in the input equals $Y$ in the output, the $2SAT$ problem must be linearly dependent.

\subsection{The linear independence of $3SAT$ problems}

The linear dependence of $2SAT$ is proved, with the same method, it seems that $3SAT$ is also linearly dependent. However, for $3SAT$ problems, because the clauses have three variables, there are two logic OR operations, so the commutative and associative laws take effect. At the same time, since there are multiple logic OR operations, the sequential instruction execution axiom of computation also takes effect. And these three axioms are in conflict with each other.

By default, the computational model evaluates clauses in left-to-right order.

In the SAT problem, the unsatisfiability of the $CNF$ formula maps to $y=0,y \in Y$.

In Tarjan's algorithm, the unsatisfiability of $CNF$ formula is mapped to the fact that the true value and false value of a bit variable can not be in a strongly connected cycle at the same time, that is, there is no circle containing both the true value and false value of a variable.

We construct a simple $3SAT$ problem based on the $2SAT$ problem. Let its $CNF$ formula be:

\begin{equation}\label{cnf.3sat}
f(x_1,x_2,x_3)= x_1 \wedge x_2 \wedge ( \bar{x}_1 \vee \bar{x}_2 \vee x_3)
\end{equation}

Tarjan's algorithm is used to calculate the third clause, $\bar {x} _1 \vee \bar {x} _2 \vee x_3 $, on the basis of associative law, $\bar {x} _1 \vee \bar {x} _2 \vee x_3 = (\bar {x} _1 \vee \bar{x}_2 )\vee x_3$. First calculate $\bar{x}_1 \vee \bar{x}_2$, obtain unsatisfied value. At this point, the Tarjan's algorithm does not read $\vee x_3$.

According to the commutative law and associative law, $\bar{x}_1 \vee \bar{x} _2 \vee x_3 = (\bar {x}_1 \vee x_3) \vee \bar {x}_2 $. In this case,Tarjan's algorithm does not read $\vee \bar{x}_2$. The two cases here correspond to two pieces of data.

So let's write these down, and we get the matrix $\Phi'$.

\begin{equation}\label{arr}
\Phi'=
\left(\begin{array}{ccccc}	
	0 & 0 & 0 & 0 &	0\\
	0 & 0 & 1 & 0 &	0\\
	0 & 1 & 0 & 0 &	0\\
	0 & 1 & 1 & 0 &	0\\
	1 & 0 & 0 & 0 &	0\\
	1 & 0 & 1 & 0 &	0\\
	1 & 1 & 0 & 0 &	0\\
	1 & 1 & 1 & 0 &	1\\
	1 & 1 & 1 & 1 &	1\\
\end{array} \right)
\end{equation}

The data in line 8 of the matrix \ref{arr} corresponds to the case where $\vee x_3$ is not read, and the data in line 9 corresponds to the case where $\vee \bar{x}_2$ is not read.

Obviously, in this matrix,$y \neq f$, $f$ is the existence of cycle in the form of $2SAT$. Simple calculation tells us that the $3SAT$ problem is linearly independent.

Actually, in the formula \ref{cnf.3sat} the third clause of the execution, there is a other way, that is $\bar {x}_1 \vee \bar{x}_2 \vee x_3 = \bar{x}_1 \vee (\bar {x}_2 \vee x_3) $.

$\bar{x}_2 \vee x_3$ is regarded as $\bar{x}_2'$, then $\bar{x}_2'$ can be regarded as a logical variable with three values, the corresponding relationship is shown in the Table \ref{tab.valueCompare}.

\begin{table}[hbp!]
\centering
\caption{The Contrast of Logic Values}
\label{tab.valueCompare}
\begin{tabular}{ccc}
	\toprule
	$x_2$ & $x_3$ & $\bar{x}_2'=\bar{x}_2 \vee x_3$	\\
	\midrule	
	0 & 0 & 0\\ 
	0 & 1 & 1  \\
	1 & 1/0 & 1  \\
	\bottomrule
\end{tabular}
\end{table}

Under three-valued logic, $x_2$ is true, real false and virtual false (i.e.,$x_2 =0,x_3=1$). When $x_2$ is true or false, Tarjan's algorithm can work normally, that is, it can get an unsatisfiable conclusion. When $x_2$ is false, since there is no line between false node and true node, the closed loop i.e., cycle, will not be formed. Therefore, Tarjan's algorithm will give a satisfactory conclusion.

Three-valued logic shows that the $2SAT$ problem, i.e., the 2-variable 1 binary operation dimension of the $SAT$ problem, is a cycle structure; while the $3SAT$ problem, i.e., the 3-variable 2 binary operation dimension of the $SAT$ problem, can no longer be a cycle structure. This fact is similar to the fact that a circle in topology in two dimension plane does not form a circle in three dimension plane.

For kSAT problem, the third clause in the formula \ref{cnf.3sat} can be transformed. Like adding $x_3$, more $x_k$ can be added, that is, it can be transformed into the formula \ref{cnf.3sat.m}.

\begin{equation}\label{cnf.3sat.m}
f(x_1,x_2,x_3)= x_1 \wedge x_2 \wedge ( \bar{x}_1 \vee \bar{x}_2 \vee x_3 \cdots \vee x_k)
\end{equation}

$x_2$ in  formula \ref{cnf.3sat.m} can be regarded as k-valued logic. k-valued logic is similar to three-valued logic in that it does not change the conclusion that the cycle has different continuity in different dimensions, which is corresponding to the reduction from $kSAT$ problem to $3SAT$ problem.

\section{Conclusions}

This paper has proved that the nature of $P$ is a polynomial-time function on real integer, meanwhile, the nature of $NP$ is a polynomial-time function on complex style, they have different dimensions which lead to $P \not =NP$. Of course, because there are a few representations of the system or the problem or the algorithm, there are a few proving ways to achieve the same results in this paper. Currently, we have at least known the discrete representation, the systemic representation, the field number representation, the algorithmic representation, the topological representation and the machine representation. These representations will demonstrate different aspects of $P$ and $NP$.

\bibliographystyle{amsplain}
\bibliography{ref}

\providecommand{\bysame}{\leavevmode\hbox to3em{\hrulefill}\thinspace}
\providecommand{\MR}{\relax\ifhmode\unskip\space\fi MR }
\providecommand{\MRhref}[2]{%
  \href{http://www.ams.org/mathscinet-getitem?mr=#1}{#2}
}
\providecommand{\href}[2]{#2}
\begin{thebibliography}{1}

\bibitem{luo2015en}
Sanjeev Arora and Boaz Barak, \emph{Computational complexity: A modern
  approach}, China Machine Press, 2015.

\bibitem{Cook71thecomplexity}
Stephen~A. Cook, \emph{The complexity of theorem-proving procedures}, STOC,
  ACM, 1971, pp.~151--158.

\bibitem{Gold2015}
Oded Goldreich, \emph{Computational complexity: A conceptual perspective},
  National Defense Industry Press, 2015.

\end{thebibliography}
\end{document}